\documentclass[acmtos]{acmsmall}
\usepackage{times}
\usepackage{fancyheadings}
\usepackage{color}
\usepackage[utf8]{inputenc}
\usepackage{amsmath}
\usepackage{amssymb}
\usepackage{amsfonts}
\usepackage{booktabs}
\usepackage{algorithmic}
\usepackage{array}
\usepackage[font=footnotesize,caption=false]{subfig}
\usepackage{url}
\usepackage{graphicx}
\usepackage{cite}
\usepackage{ifpdf}
\ifpdf
 \usepackage{epstopdf}
\fi

\newcommand{\CO}[1]
{
}

\begin{document}

\markboth{L. Pamies-Juarez and E. Biersack}{Cost Analysis of Redundancy Schemes for Distributed Storage Systems}

\title{Cost Analysis of Redundancy Schemes for Distributed Storage Systems}
\author{Lluis Pamies-Juarez
\affil{Universitat Rovira i Virgili (Tarragona, Spain)}
Ernst Biersack
\affil{Eurecom (Sophia-Antipolis, France)}
\vspace{1cm}\textnormal{\small April 15, 2011}
}

\begin{abstract} Distributed storage infrastructures require the use of data redundancy to achieve high data
reliability.  Unfortunately, the use of redundancy introduces storage and communication overheads, which can either
reduce the overall storage capacity of the system or increase its costs. To mitigate the storage and communication
overhead, different redundancy schemes have been proposed. However, due to the great variety of underlaying storage
infrastructures and the different application needs, optimizing these redundancy schemes for each storage infrastructure
is cumbersome.  The lack of rules to determine the optimal level of redundancy for each storage configuration leads
developers in industry to often choose simpler redundancy schemes, which are usually not the optimal ones. In this paper
we analyze the cost of different redundancy schemes and derive a set of rules to determine which redundancy scheme
minimizes the storage and the communication costs for a given system configuration. Additionally, we use simulation to
show that theoretically-optimal schemes may not be viable in a realistic setting where nodes can go off-line and repairs
may be delayed. In these cases, we identify which are the trade-offs between the storage and communication overheads of
the redundancy scheme and its data reliability.\end{abstract}




\begin{bottomstuff}
Author's addresses: L. Pamies-Juarez (lluis.pamies@urv.cat) {and} E. Biersack (ernst.biersack@eurecom.fr).
\end{bottomstuff}

\maketitle


\section{Introduction}
\label{s:introduction}
Distributed storage systems are widely used today for reasons of scalability and performance. There are distributed
file-systems such as Google FS~\cite{googlefs}, HDFS~\cite{hadoopfs}, GPFS~\cite{gpfs} or Dynamo~\cite{dynamo} and
peer-to-peer (P2P) storage applications like Wuala~\cite{wuala} or OceanStore~\cite{oceanstore}.

\CO{, behind cloud-storage services like Amazon's S3~\cite{s3}, or in edge-storage~\cite{echos} systems like
Cleversafe~\cite{cleversafe} or Farsite~\cite{farsite}.  There are several reasons why distributed storage mechanisms
are preferred over centric-based solutions including scalability, performance, and reliability. }

To achieve high reliability in distributed storage systems, a certain level of data redundancy is required.
Unfortunately, the use of redundancy increases the storage and communication costs of the system: (i) the space required
to store each file is increased, and (ii) additional communication bandwidth is required to repair lost data. It is
important to optimize redundancy schemes in order to minimize these storage and communication costs. For example, in
data centers where the energy cost associated with the storage sub-system represents about 40\% of the energy
consumption of all the IT components~\cite{energy}, minimizing storage cost can significantly reduce the per-byte cost
of the storage system. Whereas in less-reliable infrastructures |i.e. P2P systems| where the storage capacity is mainly
constrained by the cross-system communications bandwidth~\cite{blakerodrigues}, minimizing communication costs can
increase the overall storage capacity of the system.

Different redundancy schemes have been proposed to minimize the storage and communication costs associated with
redundancy. Redundancy schemes based on coding techniques such as Reed-Solomon codes~\cite{reedsolomon} or
LDPC~\cite{lpdcp2p} allow to achieve significant storage savings as compared to simple
replication~\cite{highavail,ecreplication,codingvsrepl,regeneratingcodes}. Moreover, recent advances in network coding
have lead to the design of Regenerating Codes~\cite{regeneratingcodes} that allow to reduce both, the storage and
communication costs, as compared to replication. While coding schemes can provide cost efficient redundancy in
production environments~\cite{hadoopec,ecgoogle,diskreduce}, distributed storage designers are still slow in adapting
advanced coding schemes for their systems. In our opinion, one reason for this reluctance is that coding schemes present
too many configuration trade-offs that make it difficult to determine the optimal configuration for a given storage
infrastructure.

Besides coding or replication one can also combine these two techniques into a hybrid redundancy scheme. In some
circumstances these hybrid redundancy schemes can reduce the costs of coding schemes~\cite{redschemes,glacier}. Besides
reducing costs, there are other reasons why maintaining whole file replicas in conjunction with encoded copies is
advantageous: (i) production systems using replication that want to reduce their costs without migrating their whole
infrastructure, (ii) peer-assisted cloud storage systems~\cite{assistedbackup}, like Wuala~\cite{wuala} that aim to
reduce the outgoing cloud bandwidth by combining cloud-storage with P2P storage, and (iii) storage systems needing
efficient file retrievals that cannot afford the computational costs inherent in coding schemes.  Unfortunately, there
are no studies analyzing under which conditions |i.e. node dynamics and network parameters|  hybrid schemes can reduce
the storage and communication costs as compared to simple replication.

Due to the great variety of redundancy schemes, it is complex to determine which redundancy scheme is the best for a
given infrastructure that is defined by properties like size (number of storage nodes), amount of stored data, node
dynamics, and cross-system bandwidth. The aim of this paper is to analyze the impact of different properties on the
storage and communication costs of the redundancy scheme.  We focus our analysis on Regenerating
Codes~\cite{regeneratingcodes}. As we will see in Section~\ref{s:regenerating}, Regenerating Codes provide a generic
framework that also allows us to analyze replication schemes and maximum-distance separable (MDS) codes such as
Reed-Solomon codes as specific instances  of Regenerating Codes.

The main contributions of our paper are as follows:
\begin{itemize}
\item This paper is the first to completely evaluate the storage and communication costs of  Regenerating Codes under
different system conditions.

\item For storage systems that need to maintain whole replicas of the stored files, we identify the conditions where a
hybrid scheme (replication+coding) can reduce the storage and communication costs of a simple replication scheme. 

\item Finally, we evaluate through simulation the effects that different redundancy scheme configurations have on the
scalability of the storage system. We show that some theoretically-optimal schemes cannot guarantee data reliability in
realistic storage environments.
\end{itemize}

The rest of the paper is organized as follows. In Section~\ref{s:relatedwork} we present the related work. In
sections~\ref{s:model} and~\ref{s:regenerating} we describe our storage model and Regenerating Codes.  In
Section~\ref{s:costs}, we analytically evaluate the storage and communication costs of Regenerating Codes. In
Section~\ref{s:hybrid} we analyze a hybrid redundancy scheme that combines Regenerating Codes and replication. Finally,
in Section~\ref{s:experiment} we validate and extend our analytical results using simulations, and in
Section~\ref{s:conclusions}, we state our conclusions.

\section{Related Work}
\label{s:relatedwork}

Tolerating node failures is a key requirement to achieve data reliability in distributed storage systems. Existing
distributed storage systems use different strategies to cope with these node failures depending on whether these
failures are transient |nodes reconnect without losing any data| or permanent |nodes disconnect and lose their data.  In
this section we present the existing techniques used to alleviate the costs caused by these transient and permanent node
failures.

Transient node failures cause temporal data unavailabilities that may prevent users from retrieving their stored files.
To tolerate transient node failures and guarantee high data availability, storage systems need to introduce data redundancy.
Redundancy ensures (with high probability) that files can be retrieved even when some storage nodes are temporally
off-line. The simplest way to introduce redundancy is by replicating each stored file. However, redundancy schemes based
on coding techniques can significantly reduce the amount of redundancy (and storage space) required while achieving the
same data reliability~\cite{codingvsrepl,replicastrategies}. Lin et al.~\cite{ecreplication} showed that such a
reduction in redundancy is only possible  when node on-line availabilities are high. For example, nodes must be more than 50\%
of the time on-line when files are stored occupying twice their original size, or more than 33\% of the time on-line
when files occupy three times their original size.

To cope with  permanent node failures, storage systems need to repair the lost redundancy. Unfortunately, repairing such
lost redundancy introduces communication overheads since it requires to move large amounts of data between nodes.  Blake
and Rodrigues~\cite{blakerodrigues} demonstrated that the communication bandwidth used by these repairs can limit the
scalability of the system in three main situations: (i) when the node failure rate is high, (ii) when the cross-system
bandwidth is low, (iii) or when the system stores too much data. Additionally, Rodrigues and Liskov~\cite{highavail}
compared replication and erasure codes in terms of communications overheads and concluded that when on-line node
availabilities are high, replication requires less communication than erasure codes.  These results pose a dilemma for
storage designers: \textit{ when node on-line availabilities are high, erasure codes minimize storage
overheads~\cite{ecreplication} and replication minimize communication overheads~\cite{highavail}}.

In order to reduce communication overheads for erasure codes, Wu et al.~\cite{redschemes} proposed the use of a hybrid
scheme combining erasure codes and replication. Although this technique slightly increases the storage overhead, it can
significantly reduce the communication overhead of erasure codes when node on-line availabilities are high. Another
technique used to minimize the communication overhead consists in lazy redundancy
maintenance~\cite{totalrecall,storagechurn} which amortizes the costs of several consecutive repairs. However, deferring
repairs can reduce the amount of available redundancy, requiring extra redundancy to guarantee the same data
reliability. Furthermore, lazy repairs lead to spikes in the network resource
usage~\cite{proactiveestim,proactivereplication}.

New coding schemes such as Hierarchical Codes or tree structured data regeneration have also been proposed to reduce the
communication overhead as compared to classical erasure codes~\cite{treeregenerating,hierarchicalcodes}. These solutions
propose storage optimizations that exploit heterogeneities in node bandwidth and node availabilities. \CO{
In~\cite{hierarchicalcodes,selfrepairing} communication costs are reduced by proposing asymmetric repairs: the algorithm
used to repair lost nodes in not the same for all nodes. However, exploiting node heterogeneities and asymmetrical
repairs introduce more difficulties to storage designers and complicates the choice of optimal redundancy
configurations.} Finally, Dimakis et~al. presented Regenerating Codes~\cite{regeneratingcodes} as a flexible redundancy
scheme for distributed storage systems. Regenerating Codes use ideas from network coding to define a new family of
erasure codes that can achieve different trade-offs in the optimization of storage and communication costs. This
flexibility allows to adjust the code to the underlaying storage infrastructure. However, there are no studies on how
Regenerating Codes should be adapted to these infrastructures, or how Regenerating Codes should be configured when
combined with file replication in hybrid schemes. In this paper we will use Regenerating
Codes~\cite{regeneratingcodes,regeneratingsurvey} as the base of our analysis on how to adapt and optimize redundancy
schemes for different underlying storage infrastructures and  different application needs.

\section{Modeling a Generic Distributed Storage System}
\label{s:model}

We consider a storage system where nodes dynamically join and leave the system~\cite{proactiveestim,datadurab}. We
assume that node lifetimes are random and follow some specific distribution $L$. Because of these dynamics, the number
of on-line nodes at time $t$, $N_t$, is a random process that fluctuates over time. Once stationarity is reached, we can
replace $N_t$ by its limiting version $N=\lim_{t\to\infty}N_t$.  Assuming that node arrivals follow a Poisson process
with a constant rate $\lambda$, then the average number of nodes in the system is
$N=\lambda\cdot\mathbb{E}[L]$~\cite{noavailability}. Additionally, it has been observed in real traces that during their
lifetime in the system, nodes present several off-line periods caused by transient
failures~\cite{globalkad,guha06experimental}. To model these transient failures, we model each node as an alternating
process between on-line and off-line states. The sojourn times at these states are random and follow two different
distributions: $\mathcal{X}_\text{on}$ and $\mathcal{X}_\text{off}$ respectively. Using these distributions we can
measure the node on-line availability in stationary state as~\cite{heterogeneouschurn}: $$a =
\frac{\mathbb{E}[\mathcal{X}_\text{on}]}{\mathbb{E}[\mathcal{X}_\text{on}]+\mathbb{E}[\mathcal{X}_\text{off}]}.$$

All the $N$ nodes in the system are responsible to store a constant amount of data that is uniformly distributed among
the $N$ nodes. To model different data granularity, we will consider that this total amount of stored data corresponds
to $O$ different data files of size $\mathcal M$ bytes. However, since each of these files is stored with redundancy,
the total disk space required to store each file is $R\cdot\mathcal M$, being $R$ the redundancy factor (or stretch
factor). The value of $R$ is set to guarantee that files are always available with a probability, $p$, that is very
close to one.

When a node reaches the end of its life, it abandons the system, losing all the data stored on it. A repair process is
responsible to recreate the lost redundancy and to ensure that the retrieve probability, $p$, is not compromised. There
are three main approaches used to recreate redundancy when nodes fail:
\begin{enumerate}
\item {\bf Eager repairs}: Lost redundancy is repaired on demand immediately after a node failure is detected.
\item {\bf Lazy repairs}: The system waits until a certain number of nodes had failed and repairs them all at once.
\item {\bf Proactive repairs}: The system schedules the insertion of new redundancy at a constant rate, which is set
according to the average node failure rate.
\end{enumerate}
In our storage model we will assume the use of \textit{proactive repairs}. Compared to eager repairs, proactive repairs
simplify the analysis of the communication costs. Furthermore, while  lazy repair   can reduce the maintenance costs by
amortizing the communication costs across several repairs~\cite{storagechurn}, it presents some important drawbacks: (i)
delaying repairs leads to  periods with low-redundancy that makes the system vulnerable; (ii) lazy repairs cause
network resource usage to occur in bursts, creating spikes of system utilization~\cite{proactiveestim}. By adapting a
proactive repair strategy, communication overheads are evenly distributed in time  and we can analyze the storage system
in its steady state~\cite{proactiveestim}.

\section{Regenerating Codes}
\label{s:regenerating}

Regenerating Codes~\cite{regeneratingcodes} are a family of erasure codes that allow to trade-off communication cost for
storage cost and vice versa. To store a file of size $\mathcal M$ bytes, Regenerating Codes generate $n$ blocks each to
be stored on a different storage node. Each of these \textbf{storage blocks} has a size of $\alpha$, which makes the
file \textbf{stretch factor} $R$ to be $R=n\alpha/\mathcal M$. When a storage node leaves the system or when a failure
occurs, a new node can repair the lost block by downloading a \textbf{repair block} of size $\beta$ bytes, $\beta\leq\alpha$,
from any set of $d$ out of $n-1$ alive nodes ($k \leq d\leq n-1$). We will refer to $d$ as the \textbf{repair
degree}. The total amount of data received by the repairing node, $\gamma$, $\gamma=d\beta$, is called the \textbf{repair
bandwidth}. Finally, a node can \textit{reconstruct} the file by downloading $\alpha$ bytes (the entire storage block)
from $k$ different nodes. In Figure~\ref{f:infoflow} we depict the basic operations of file retrieve and block repair
for a Regenerating Code . The labels at the edges indicate to the amount of data transmitted between nodes during each
operation.

\begin{figure}
\centering
\includegraphics[width=.55\textwidth]{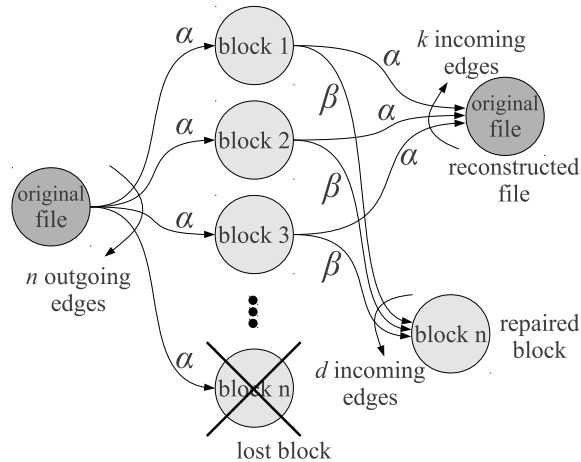}
\caption{Scheme for the repair and retrieve operations of Regenerating Codes.}
\label{f:infoflow}
\end{figure}

Dimakis et~al.~\cite{regeneratingcodes} gave the conditions that the set of parameters $(n,k,d,\alpha,\gamma=d\beta)$
must satisfy to construct a valid Regenerating Code. Basically, once the subset of parameters: $(n,k,d)$ is fixed,
Dimakis~et~al. obtained an analytical expression for the relationship between the values of $\alpha$ and $\gamma$.  This
$\alpha$-$\gamma$ relationship presents a trade-off curve: the larger $\alpha$, the smaller $\gamma$, and vice-versa. It
means that it is impossible to simultaneously minimize both, communication cost and storage cost. Since the maximum
storage capacity of the system can be constrained either by bandwidth bottlenecks or disk storage bottlenecks, there are
two extreme ($\alpha,\gamma$)-points of this trade-off curve that are of special interest w.r.t.~maximizing the storage
capacity.  The first is the point where the storage block size $\alpha$ per node is minimized, which is referred to as
\textbf{ Minimum Storage Regenerating} (MSR) code.  The second is the point where the repair bandwidth $\gamma$ is
minimized, which is referred to as \textbf{Minimum Bandwidth Regenerating} (MBR) code.  According
to~\cite{regeneratingsurvey}, the storage block size ($\alpha$) and the repair bandwidth ($\gamma$) for MSR and MBR
codes are:

\begin{align}
\left( \alpha_\text{MSR},\gamma_\text{MSR}\right) &= \left(\frac{\mathcal M}{k},~\frac{\mathcal
M}{k}\frac{d}{d-k+1}\right)\label{e:parmsr}\\
\left( \alpha_\text{MBR},\gamma_\text{MBR}\right) &= \left(\frac{\mathcal M}{k}\frac{2d}{2d-k+1},~\frac{\mathcal
M}{k}\frac{2d}{2d-k+1}\right)\label{e:parmbr}
\end{align}

There are two particular MSR configurations of special interest:
\begin{itemize}
\item \emph{Maximum-distance separable (MDS) codes:} In MSR codes, when $d=k$, we obtain
$\beta_\text{MSR}=\alpha_\text{MSR}$ and the Regenerating Code behaves exactly like a traditional MDS codes such as a
Reed Solomon code~\cite{reedsolomon}. In this case, the repair bandwidth, $\gamma_\text{MDS}\left[k=d\right]$, is
identical to the size of the original file, $\mathcal M$:
$$\gamma_\text{MDS}\left[k=d\right]=d~\beta_\text{MSR}=k~\alpha_\text{MSR}=k~\frac{\mathcal M}{k}=\mathcal M.$$

\item \emph{File replication:} In MSR codes, when $k=d=1$, the code becomes a simple replication scheme where the $n$
storage nodes each store a complete copy of the original file. For $k=d=1$, the storage block size,
$\alpha_\text{MSR}\left[k=d=1\right]$, and the repair bandwidth, $\gamma_\text{MSR}\left[k=d=1\right]$, are equal to the
size of the original file, $\alpha_\text{MSR}\left[k=d=1\right]=\gamma_\text{MSR}\left[k=d=1\right]=\mathcal M$.
\end{itemize}

\begin{table}
\tbl{Symbols used.\label{t:symbols}}{%
\centering
\begin{tabular}{|c|l|}\hline
$N$ & Average number of storage nodes. \\ \hline
$\lambda$ & Node arrival/departure rate (nodes/sec.). \\ \hline
$L$ & Distribution of the node lifetime (sec.). \\ \hline
$a$ & Node on-line availability. \\ \hline
$O$ & Number of stored files. \\ \hline
$\mathcal M$ & Size of the stored files (bytes).\\ \hline
$\omega$ & Service bandwidth of each node (KBps).\\ \hline
$p$ & Data availability. Probability of detecting $k$ blocks on-line.\\ \hline
$n$ & Number of storage blocks.\\ \hline
$k$ & Retrieval degree: number of blocks required for retrieval of original data.\\ \hline
$d$ & Repair degree: number of blocks required for repair.\\ \hline
$\alpha$ & Storage block size.\\ \hline
$\beta$ & Repair block size.\\ \hline
$\gamma$ & Repair bandwidth.\\ \hline
\end{tabular}}
\end{table}
In Table~\ref{t:symbols} we summarize the symbols used throughout the paper.

\section{Cost Analysis}
\label{s:costs}

\subsection{Redundancy Cost}
\label{s:redcosts}

In Section~\ref{s:regenerating} we defined \textbf{data redundancy}  as $R=n\alpha/\mathcal M$.
In this section we aim to measure the minimum $R$ required to guarantee a desired \textbf{file retrieve probability} $p$.
Since in Regenerating Codes the retrieval process needs to download $k$ different blocks out of the total $n$ blocks,
the retrieve probability $p$ is measured as~\cite{ecreplication},
\begin{equation}
p = \sum_{i=k}^{n} \binom{n}{i} a^i(1-a)^{n-i}.
\label{e:durab}
\end{equation}
Given the values of $k$, $a$ and $p$, we can use eq.~(\ref{e:durab}) to determine the minimum number of redundant blocks
required to guarantee a certain retrieve probability $p$ using the function $\eta$:
\begin{equation}
\eta[k,a,p]=\text{min}\left\{n':~~p \leq \sum_{i=k}^{n'} \binom{n'}{i} a^i(1-a)^{n'-i},~n'\geq k\right\}.
\label{e:n}
\end{equation}
Note that the number of redundant blocks required to achieve $p$ is a function of the repair degree, $k$, the node
on-line availability, $a$, and $p$. In the rest of this paper we will use the notation $\eta[k,a,p]$ to refer to the
number of storage blocks $n$ required to achieve a retrieve
probability $p$ for the specific $k$ and $a$ values.  

Since data redundancy is $R=n\alpha/\mathcal M$, we can obtain the redundancy required by MSR and MBR codes,
$R_\text{MSR}$ and $R_\text{MBR}$ respectively, by substituting
$\alpha$ with the expressions given for $\alpha$  in
eq.~(\ref{e:parmsr}) and eq.~(\ref{e:parmbr}):
\begin{align}
R_\text{MSR} &= \frac{\eta[k,a,p]\cdot\alpha_\text{MSR}}{\mathcal M} = \frac{\eta[k,a,p]\cdot\left(\mathcal
M/k\right)}{\mathcal M} = \boxed{\frac{\eta[k,a,p]}{k}} \label{e:red_cod_msr}\\
R_\text{MBR} &= \frac{\eta[k,a,p]\cdot\alpha_\text{MBR}}{\mathcal M} = \frac{\eta[k,a,p]\cdot\left(2d\mathcal
M/(k(2d-k+1))\right)}{\mathcal M} = \boxed{\frac{2d\cdot \eta[k,a,p]}{k(2d-k+1)}}
\label{e:red_cod_mbr}
\end{align}

Using these expressions we can state the following lemma:
\begin{lemma}
For $n$, $k$ and $d$ fixed, the redundancy $R_\text{MSR}$, required by MSR codes is always smaller than or equal
to the redundancy $R_\text{MBR}$  required by MBR codes.
\label{l:rmsr}
\end{lemma}
\begin{proof}
We can state the lemma as $R_\text{MSR}\leq R_\text{MBR}$. Using equations~(\ref{e:red_cod_msr})
and~(\ref{e:red_cod_mbr}) we obtain:
\begin{align*}
\frac{\eta[k,a,p]\cdot~\alpha_\text{MSR}}{\mathcal M} &\leq \frac{\eta[k,a,p]\cdot~\alpha_\text{MBR}}{\mathcal M} \\
\alpha_\text{MSR} &\leq \alpha_\text{MBR},
\end{align*}
which is true by the definition of MSR codes and MBR codes~\cite{regeneratingcodes}.
\end{proof}

\begin{figure}
\centering
\subfloat[Redundancy for MSR codes (includes MDS codes).]{\label{f:redund:msr}\includegraphics{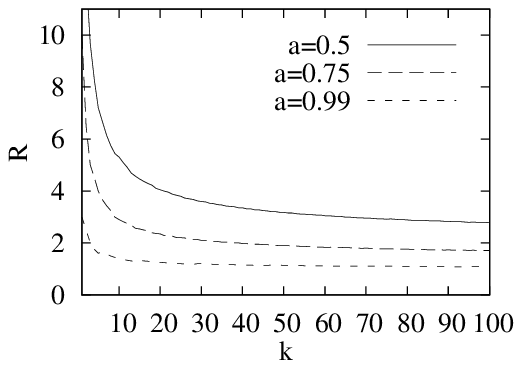}}
\subfloat[Redundancy for MBR codes.]{\label{f:redund:mbr}\includegraphics{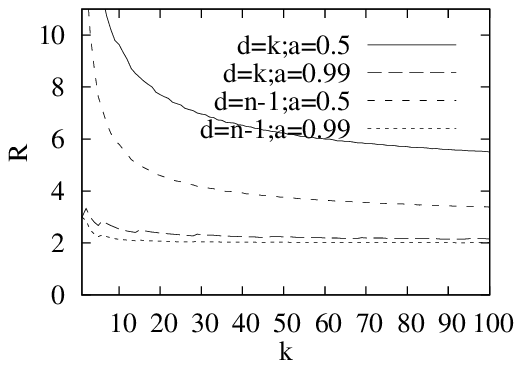}}
\quad
\subfloat[Value of {$\eta[k,a,0.999999]$}.]{\label{f:redund:n}\includegraphics{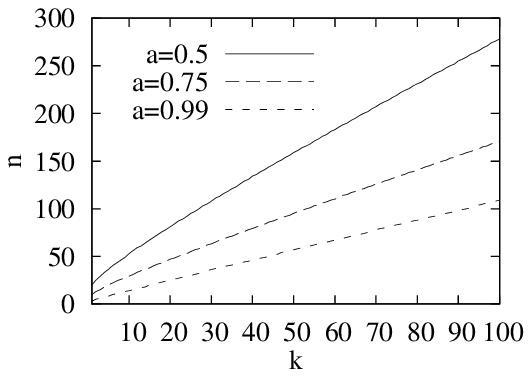}}
\caption{Redundancy $R$ required to achieve a retrieve probability $p=0.999999$ for MSR and MBR codes as a function of
the retrieve degree $k$. Each plot in (a) and (b) depicts the redundancy evaluated using eq.~(\ref{e:red_cod_msr}) and
eq.~(\ref{e:red_cod_mbr}) for different values of  $d$, and different values of the node on-line availability
$a$. In (c) we plot the number of storage blocks $n$ required to achieve the retrieve probability $p$ for each case.}
\label{f:redund}
\end{figure}

In Figure~\ref{f:redund:msr} and \ref{f:redund:mbr} we plot the redundancy $R$ required to achieve a retrieve
probability $p=0.999999$ for MSR   and MBR codes. We plot the values of $R$ as a function of the
retrieve degree, $k$, and for different node availabilities, $a$. Additionally, for MBR codes we also depict the values
of $R_\text{MBR}$ for the two extreme repair degree values: $d=k$ and $d=n-1$. We do not evaluate $R_\text{MSR}$ for
different $d$ values because $R_\text{MSR}$ is independent of $d$ (see eq.~(\ref{e:red_cod_msr})). In
Figure~\ref{f:redund:n} we use eq.~(\ref{e:n}) to plot the number of blocks, $\eta[k,a,p]$, used in
figures~\ref{f:redund:msr} and \ref{f:redund:mbr} for the retrieve probability $p=0.999999$.

In Figure~\ref{f:redund} we can see that for MSR and MBR, increasing $k$ reduces $R$, and therefore, reduces storage
costs. Additionally, comparing figures~\ref{f:redund:msr} and~\ref{f:redund:mbr} we can appreciate the consequences of
Lemma~\ref{l:rmsr}: for a given node availability, $a$, and a retrieve degree $k$, the redundancy required for MSR codes
is always smaller than the redundancy required for MBR codes. Finally, we can see that $R$ first quickly deceases with
increasing $k$ before it reaches its asymptotic values. There is no point in choosing $k$ very large to minimize the
storage costs of MSR and MBR codes, since large $k$ values induce a very high computational cost for coding and
decoding~\cite{practicalregenerating}. We therefore recommend to use values for $k$ where the redundancy $R$ starts
approaching the asymptote, namely $k=5$ for $a=0.99$, $k=20$ for $a=0.75$ and $k=50$ for $a=0.5$. In
Table~\ref{t:savings} we provide the redundancy savings achieved by using these $k$ values.

\begin{table}
\tbl{Storage space savings for adopting a Regenerating Code  instead of  replication. We use different $k$
values for each on-line node availability and a target retrieve probability of $p=0.999999$.}{%
\centering
\begin{tabular}{|l|c|c|c|} \hline
\parbox{2cm}{\vspace{.1cm}\vspace{.1cm}} & 
\parbox{2cm}{\vspace{.1cm}\centering $a=0.5$\\$k=50$\vspace{.1cm}} & 
\parbox{2cm}{\vspace{.1cm}\centering $a=0.75$\\$k=20$\vspace{.1cm}} &
\parbox{2cm}{\vspace{.1cm}\centering $a=0.99$\\$k=5$\vspace{.1cm}} \\ \hline
MSR & 47\% & 77\% & 84\% \\ \hline
MBR ($d=k$) & 69\% & 55\% & 11\% \\ \hline
MBR ($d=n-1$) & 81\% & 70\% & 25\% \\ \hline
\end{tabular}}
\label{t:savings}
\end{table}

\subsection{Communication Costs}

When a node fails, the system must repair all the data blocks stored on the failed node.  Repairing each of these blocks
requires to transfer data between nodes, which entails a communication cost. In this section we measure the minimum
per-node bandwidth required to sustain the overall repair traffic of the storage system.  We will first compute the
total amount of data that is transfered within the system during a period of time $\Delta$:
\begin{align}
\text{data transfered during }\Delta ~= ~~&\text{nodes failed during }\Delta~~\times~~\text{blocks stored per
node}~~\times\label{e:deltatraff}\\\nonumber\times~~&\text{traffic to repair one block}.
\end{align}

Assuming that there are $N$ nodes with an average lifetime $\mathbb{E}[L]$, the average number of nodes that fail during a
period $\Delta$ is $\Delta N/\mathbb{E}[L]$~\cite{proactiveestim}. Additionally, assuming that data blocks are uniformly
distributed between all storage nodes, the average number of blocks stored per node is $n\cdot O/N$. Finally, since the
traffic required to repair one failed block is $\gamma$, we can rewrite eq.~(\ref{e:deltatraff}) as:
\begin{equation}
\text{data transfered during }\Delta = \left(\Delta~\frac{N}{\mathbb{E}[L]}\right)\times \left(\frac{n~O}{N}\right) \times
\gamma
\label{e:deltatraff2}
\end{equation}

Then, the minimum per-node bandwidth, $W$, required to ensure that all stored data can be successfully repaired is the ratio between
the amount of data transmitted per unit of time (in seconds), and the average number of on-line nodes, $aN$:
\begin{equation}
W = \frac{\text{data transfered during }\Delta}{\Delta\times\text{avg. number on-line nodes}} =
\frac{\gamma~n~O}{a~N~\mathbb{E}[L]}.
\label{e:bw}
\end{equation}
Assuming that the repair bandwidth, $\gamma$, is given in KB, and the node lifetime, $L$, in seconds, then the minimum
per-node bandwidth $W$ is expressed in KBps. Assuming that the upload bandwidth of each node is always smaller than or
equal to the download bandwidth, this minimum per-node bandwidth, $W$, represents the minimum upload bandwidth required
by each node.

If we use the values of the repair bandwidth $\gamma$ given in equations~(\ref{e:parmsr}) and~(\ref{e:parmbr}), we
obtain the minimum per-node bandwidth for each Regenerating Code configuration:
\begin{align}
W_\text{MSR} &= \gamma_\text{MSR}\cdot\frac{\eta[k,a,p]~O}{a~N~\mathbb{E}[L]} = \frac{\mathcal
M}{k}\frac{d}{(d-k+1)}\frac{\eta[k,a,p]~O}{a~N~\mathbb{E}[L]} = \boxed{\frac{d\cdot
\eta[k,a,p]}{ak(d-k+1)}\frac{O~\mathcal M}{N\mathbb{E}[L]}}
\label{e:WMSR}\\
W_\text{MBR} &= \gamma_\text{MBR}\cdot\frac{\eta[k,a,p]~O}{a~N~\mathbb{E}[L]} = \frac{\mathcal
M}{k}\frac{2d}{(2d-k+1)}\frac{\eta[k,a,p]~O}{a~N~\mathbb{E}[L]} = \boxed{\frac{2d\cdot
\eta[k,a,p]}{ak(2d-k+1)}\frac{O~\mathcal M}{N\mathbb{E}[L]}}
\label{e:WMBR}
\end{align}

Taking these two expressions we can state the following lemma:
\begin{lemma}
For the same $n$, $k$ and $d$ parameters, the per-node bandwidth required by MBR codes, $W_\text{MBR}$, is always smaller
than or equal to the per-node bandwidth required by MSR codes, $W_\text{MSR}$.
\label{l:bwmbr}
\end{lemma}
\begin{proof}
We can state the lemma as $W_\text{MBR}\leq W_\text{MSR}$. Using equations~(\ref{e:WMSR})
and~(\ref{e:WMBR}) we obtain:
\begin{align*}
\gamma_\text{MBR}\cdot\frac{\eta[k,a,p]~O}{a~N~\mathbb{E}[L]} &\leq
\gamma_\text{MSR}\cdot\frac{\eta[k,a,p]~O}{a~N~\mathbb{E}[L]} \\
\gamma_\text{MBR} &\leq
\gamma_\text{MSR},
\end{align*}
which is true by the definition of MSR codes and MBR codes from~\cite{regeneratingcodes}.
\end{proof}

In the rest of this section we analyze the per-node bandwidth requirements, $W$, for MSR and MBR codes. Since in
eq.~(\ref{e:WMSR}) and eq.~(\ref{e:WMBR}) the term $\frac{O\mathcal M}{N\mathbb{E}[L]}$ does not depend on the
Regenerating Code parameters, $n, k, d$, we will assume that $\frac{O\mathcal M}{N\mathbb{E}[L]}=1$.  To obtain the
minimum per-node bandwidth, we simply have to multiply $W$ times $\frac{O\mathcal M}{N \mathbb{E}[L]}$.

\begin{figure}
\centering
\subfloat[MSR when $d=k$ and MDS.]{\label{f:bwenc:msrlow}\includegraphics{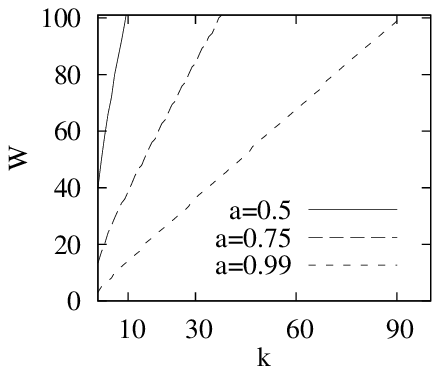}}
\subfloat[MSR when $d=n-1$.]{\label{f:bwenc:msrhigh}\includegraphics{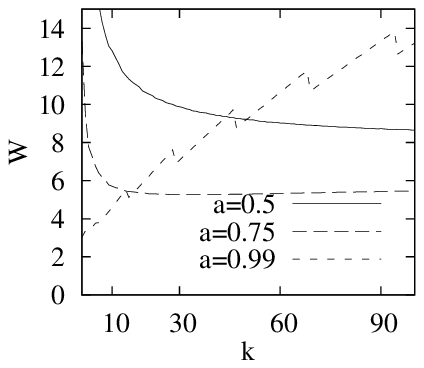}}
\caption{We use eq.~(\ref{e:WMSR}) to show the per-node bandwidth required to achieve $p=0.999999$
for MSR codes.}
\label{f:bwenc}
\end{figure}

\paragraph{Communication Cost for MSR Codes} In Figure~\ref{f:bwenc} we use eq.~(\ref{e:WMSR}) to analyze the per-node
bandwidth requirements of MSR codes when the required retrieve probability is $p=0.999999$. We plot the results for
$d=k$ and $d=n-1$ and for three different on-line node availabilities:
\begin{itemize}
\item 
For $d=k$ we can see in Figure~\ref{f:bwenc:msrlow} how the per-node bandwidth of a MDS code such as a Reed-Solomon
code, is linear in $k$. In this case, the lowest per-node bandwidth is achieved when $k=1$, which corresponds to a
simple replication scheme.
\item
For $d=n-1$, however, we can see in Figure~\ref{f:bwenc:msrhigh} that the per-node bandwidth is asymptotically
decreasing in $k$.  However, as already said, we recommend to choose $k=20$ when $a=0.75$ and $k=50$ when $a=0.5$.
Finally, we can see that for $a=0.99$, $W_\text{MBR}$ is not an asymptotically decreasing function: As $a$ tends to one,
the number of required blocks, $\eta[k,a,p]$, tends to $k$ (see eq.~(\ref{e:n})) and  the case $d=n-1$  is identical to
the case $d=k$, which is depicted in sub-figure~\ref{f:bwenc:msrlow}.
\end{itemize}

In Figure~\ref{f:bwenc:msrlow} we saw that MDS codes ($d=k;k>1$) do not reduce the per-node bandwidth as compared to
replication ($d=k=1$) while in Figure~\ref{f:bwenc:msrhigh} we saw that for $d>k$, a MSR code can reduce the bandwidth
as compared to replication except for high node on-line availabilities ($a=0.99$). We now want to determine the maximum
node on-line availability, $a$, for which a MSR code can reduce the per-node bandwidth requirement as compared to
replication. Let us denote by $W_\text{MSR}[k=d=1]$ the per-node bandwidth required by replication and
$W_\text{MSR}[k>1,d\geq k]$ denote the per-node bandwidth required by a MSR code. Then, a MSR reduces the
bandwidth required by replication when the following inequality holds:
\begin{equation}
W_\text{MSR}[k=d=1]\geq W_\text{MSR}[k>1,d\geq k]  
\label{e:Wbw}
\end{equation}

\begin{table}
\tbl{Minimum $d$ values to construct MSR codes that requiring less repair bandwidth than simple
file replication. The target retrieve probability is $p=0.999999$.\label{t:ds}}{%
\centering
\begin{tabular}{|c|c|c|c|}\hline
 & \multicolumn{3}{c|}{minimum repair degree satisfying~eq.~(\ref{e:Wbw}) and the value of $n$.}\\ \hline
Node availability & $k=50$ & $k=20$ & $k=5$ \\ \hline
$a=0.5$ & $n=159;d=59$ & $n=81;d=24$ & $n=36;d=7$ \\ \hline
$a=0.75$ & $n=95;d=61$ & $n=47;d=25$ & $n=20;d=7$ \\ \hline
$a=0.9$ & $n=71;d=65$ & $n=34;d=27$ & $n=13;d=8$ \\ \hline
$a=0.92$ & $n=69;d=64$ & $n=32;d=26$ & $n=12;d=7$ \\ \hline
$a=0.95$ & $n=64;d=--$ & $n=29;d=27$ & $n=11;d=8$ \\ \hline
$a=0.97$ & $n=61;d=--$ & $n=27;d=--$ & $n=10;d=9$ \\ \hline
$a=0.99$ & $n=57;d=--$ & $n=25;d=--$ & $n=8;d=--$ \\ \hline
\end{tabular}}
\end{table}

Table~\ref{t:ds} shows the minimum $d$ that satisfies the inequality defined in eq.~(\ref{e:Wbw}) for different on-line
node availabilities, $a$, and different retrieve degrees $k$. We additionally provide the number of storage blocks, $n$,
required to achieve $p=0.999999$. We can see  that for low node availabilities small values of $d$, slightly larger than
$k$, are sufficient to reduce the per-node bandwidth required by replication.  However, for high on-line node
availabilities, the minimum value of $d$ satisfying  eq.~(\ref{e:Wbw}) becomes larger than $n-1$, which is not a valid
Regenerating Code configuration.  This maximum on-line availability becomes higher for low $k$ values, namely $a\geq
0.95$ for $k=50$, $a\geq 0.97$ for $k=20$ and $a\geq 0.99$ for $k=5$.  We can generally state that for high on-line node
availabilities, \textit{replication becomes more bandwidth efficient than any MSR code}, which  confirms the result
obtained by Rodrigues and Liskov in~\cite{highavail}.

\begin{figure}
\centering
\subfloat[MBR when $d=k$.]{\label{f:bwenc:mbrlow}\includegraphics{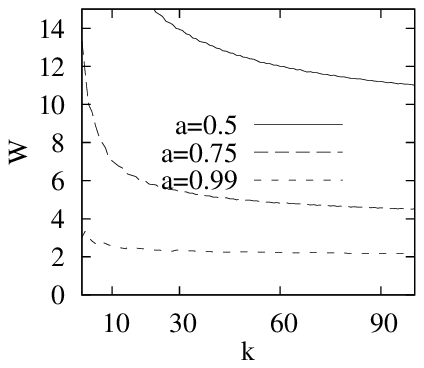}}
\subfloat[MBR when $d=n-1$.]{\label{f:bwenc:mbrhigh}\includegraphics{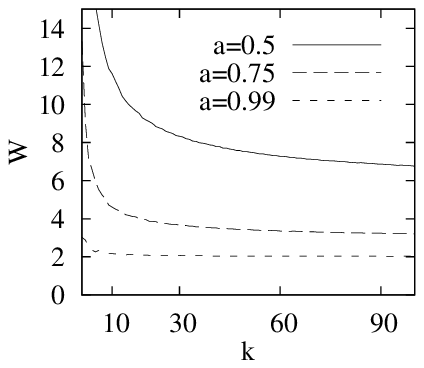}}
\caption{Per-node bandwidth required to achieve $p=0.999999$ for MBR codes using eq.~(\ref{e:bw}).}
\label{f:bwenc2}
\end{figure}

\paragraph{Communication Cost for MBR Codes} In Figure~\ref{f:bwenc2} we   plot the required per-node bandwidth of MBR
codes for  $d=k$ and $d=n-1$. For MBR codes, in difference to MSR codes,  we can see that for both $d$ values the
required per-node bandwidth  $W$  asymptotically decreases with increasing  $k$ and we can state:
\begin{remark}
For MBR codes $W_\text{MBR}\left[k=k'\right] \geq W_\text{MBR}\left[k=k'+1\right]$.
\label{r:mbrdecr}
\end{remark}
From Lemma~\ref{l:bwmbr} we know that for the same configuration, MBR codes are more bandwidth efficient than MSR
codes. Using Remark~\ref{r:mbrdecr} we can now state that all MBR codes are also more bandwidth efficient than simple
replication, which is a special case of MSR:
\begin{lemma}
The per-node bandwidth requirements of MBR codes are lower than or equal to the per-node bandwidth requirements of simple
replication: $W_\text{MBR} \leq W_\text{MSR}\left[k=d=1\right]$.
\end{lemma}
\begin{proof}
If this lemma is true, then the per-node bandwidth of the MBR configuration that consumes the most  bandwidth must be
lower than or equal to the per-node bandwidth of replication. Since $W_\text{MBR}$ is largest for $k=1$ (see
Remark~\ref{r:mbrdecr}), we can rewrite this lemma as: $W_\text{MBR}\left[k=d=1\right]\leq
W_\text{MSR}\left[k=d=1\right]$. To proof it by contradiction we assume that $W_\text{MBR}\left[k=1\right]>
W_\text{MSR}\left[k=d=1\right]$. Using equations~(\ref{e:WMSR}) and~(\ref{e:WMBR}) we obtain:
\begin{align*}
\gamma_\text{MBR}\left[k=d=1\right]\cdot\frac{\eta[1,a,p]~O}{a~N~\mathbb{E}[L]} &>
\gamma_\text{MSR}\left[k=d=1\right]\cdot\frac{\eta[1,a,p]~O}{a~N~\mathbb{E}[L]} \\
\gamma_\text{MBR}\left[k=d=1\right] &>
\gamma_\text{MSR}\left[k=d=1\right] \\
1 &> 1;
\end{align*}
which is a contradiction. 
\end{proof}

In Figure~\ref{f:replmbr} we plot the communication savings  a storage system makes when using a MBR code instead of
replication. The  savings have
the same asymptotic behavior than the bandwidth requirements, $W_\text{MBR}$, depicted in Figure~\ref{f:bwenc2}.
Since for MBR codes $\alpha_\text{MBR}=\gamma_\text{MBR}$, i.e.~the
storage block size is the same as the repair bandwidth,
the \textit{communication savings for MBR are the same as
the storage savings} listed in Table~\ref{t:savings}.

\begin{figure}
\centering
\subfloat[Savings when $d=k$.]{\label{f:bwenc:mbrlow}\includegraphics{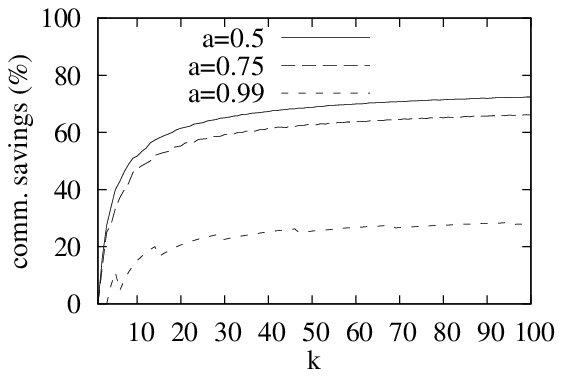}}
\subfloat[Savings when $d=n-1$.]{\label{f:bwenc:mbrhigh}\includegraphics{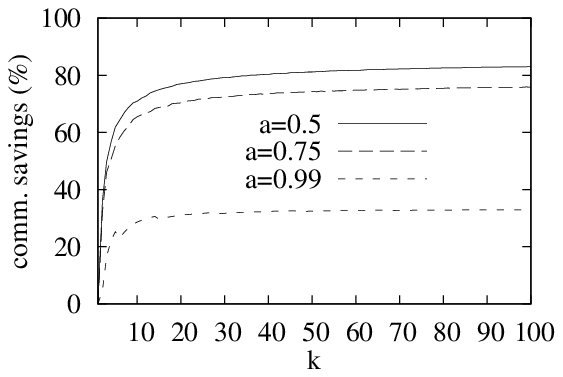}}
\caption{Reduction  of the communication cost by adopting a
MBR code instead of replication as function of $k$ a retrieve probability of $p=0.999999$.}
\label{f:replmbr}
\end{figure}

\section{Hybrid Repositories}
\label{s:hybrid}

In Section~\ref{s:costs} we saw that except for one a particular case (MSR codes and high node on-line availabilities),
MSR and MBR codes offer both, lower storage costs and lower communication costs than simple file replication. However,
there are some scenarios where the storage system needs to ensure that files can be accessed without the need of
decoding operations. For example,   storage infrastructures using replication~\cite{googlefs,hadoopfs} may not afford a
migration of their infrastructures from replication to erasure encodes. Other examples are on-line streaming services or
content distribution networks (CDNs) that need efficient access to stored files without requiring complex decoding
operations.

As we saw in Section~\ref{s:costs}, maintaining whole file replicas (MSR codes with $k=d=1$) has a higher storage cost
than using coding schemes. However, when whole file replicas are required, storage systems can reduce this high cost by
using a\textit{ hybrid redundancy scheme that combines replication and erasure codes}.  The replicas can also help
reduce the communication cost when repairing lost data by generating new redundant blocks \textit{using the on-line file
replicas}: Generating a redundant block from a file replica requires transmitting $\alpha$ bytes instead of the
$\gamma=d\cdot\beta$ bytes required by the normal repair process. From eqs.~(\ref{e:parmsr}) and~(\ref{e:parmbr}) it is
easy to see that $\alpha\leq\gamma$.  While some papers have studied hybrid redundancy
schemes~\cite{highavail,glacier,regeneratingcodes}, their aim was to reduce communication costs and not to guarantee
permanent access to replicated objects. Therefore, these papers assumed that \textit{only one} replica of each file
was maintained in the system, ignoring the two problems that arise when this replica goes temporarily off-line: (i) it
is not possible to access the file without decoding operations, and (ii) repairs using the replica are not possible.

In this section we evaluate a different hybrid scenario, where the storage system may maintain more than one replica of
the whole file in order to ensure with high probability that there is always one replica on-line. However, it is not
clear if the overall communication costs of our hybrid scheme will be lower than the communication costs of a single
replication scheme. Further, even if communication costs are reduced, the use of a double redundancy scheme (replication
and coding) may increase storage costs.  To the best of our knowledge, there is no prior work analyzing these aspects.
In our analysis we differentiate between the probability $p_\text{low}$ of having a file replica on-line, and the
retrieve probability $p$ for being able to retrieve a files using encoded blocks, which requires that $k$ out of a total
of $n$ storage blocks are on-line.  We  assume that $p_\text{low}\ll p$, for example $p_\text{low}=0.99$ and
$p=0.999999$, which is motivated by the fact that while users are likely to tolerate higher access times to a file,
which will need to be reconstructed first in some rare cases when no replicas are found on-line, but they require very
strong guarantees that  data is never lost.

\paragraph{Adapting Communication Cost to the Hybrid Scheme} In a hybrid scheme we need to consider two types of
repair traffic, namely (i) traffic $W_\text{MSR}[k=d=1]$, to repair
lost replicas and (ii) traffic $W^\text{repl}$  to repair encoded
blocks.
\CO{ We will separate
the minimum per-node bandwidth of a hybrid scheme in two parts: the bandwidth required to repair replicas,
$W_\text{MSR}[k=d=1]$, and the bandwidth required to repair encoded
blocks, $W^\text{repl}$.
}  
Since in the hybrid scheme
blocks are repaired directly from   a replicated copy, repairing an encoded block requires transmitting only
one new storage block of $\alpha$ bytes. We obtain $W^\text{repl}$ by
replacing in eq.~(\ref{e:deltatraff})   the  term ``traffic to repair a block'' in
by  $\alpha$. Arranging the terms we obtain the following two expressions:
\begin{align}
    W^\text{repl}_\text{MSR} &= \frac{\eta[k,a,p]}{ka} \times \frac{\mathcal M O}{N~\mathbb{E}[L]}\label{e:bwhybmsr}\\
    W^\text{repl}_\text{MBR} &= \frac{2d\cdot \eta[k,a,p]}{ka(2d-k+1)} \times \frac{\mathcal M
    O}{N~\mathbb{E}[L]}\label{e:bwhybmbr}
\end{align}
Note that these expressions assume that \textit{all lost blocks are repaired from replicas}. Since we are adopting a
proactive repair scheme, the system can delay individual repairs when no replicas are available. However, since replicas
are available most of the time, these delays will rarely happen.

Comparing $W^\text{repl}_\text{MSR}$, and $W^\text{repl}_\text{MBR}$
we can state  the following lemma:
\begin{lemma}
For the same $k$, $d$ and $p$ parameters, a hybrid scheme using a MBR
code has a communication cost that is at least as high as  the communication cost
of a hybrid scheme using a MSR code.
\label{l:hybbw}
\end{lemma}
\begin{proof}
We can state the lemma as $W^\text{repl}_\text{MSR}\leq W^\text{repl}_\text{MBR}$. Using equations~(\ref{e:bwhybmsr})
and~(\ref{e:bwhybmbr}) we obtain:
\begin{align*}
\frac{\eta[k,a,p]}{ka} \times \frac{\mathcal M O}{N~\mathbb{E}[L]} &\leq
\frac{2d\cdot \eta[k,a,p]}{ka(2d-k+1)} \times \frac{\mathcal MO}{N~\mathbb{E}[L]} \\
1 &\leq \frac{2d}{2d-k+1} \\
2d-k+1 &\leq 2d \\
1 &\leq k
\end{align*}
which is true by the definition of Regenerating Codes.
\end{proof}

Lemma~\ref{l:hybbw} implies that MSR codes when used in hybrid schemes
are both, more storage-efficient and more  bandwidth-efficient than
MBR codes. For this reason we will not consider the use of MBR codes
in hybrid schemes. 

Let us assume that the required retrieve
probability for the whole hybrid system is $p$ and that the retrieve probability for replicated
objects is $p_\text{low}$, $p_\text{low}\ll p$. A hybrid scheme reduces the
storage cost compared to   replication when the following
condition is  satisfied:
\begin{equation}
\underbrace{R_\text{MSR}[k=1;~p_\text{low}] + R_\text{MSR}[k>1;~p]}_\text{hybrid storage costs} <
\underbrace{R_\text{MSR}[k=1;~p]}_\text{replication storage costs}.
\label{e:limitR}
\end{equation}
And analogously, a hybrid scheme reduces communication costs when:
\begin{equation}
\underbrace{W_\text{MSR}[k=1;~p_\text{low}] + W_\text{MSR}^\text{repl}[k>1;~p]}_\text{hybrid comm. costs} <
\underbrace{W_\text{MSR}[k=1;~p]}_\text{replication comm. costs}.
\label{e:limitW}
\end{equation}

\begin{table}
\tbl{Number of replicas required to achieve a retrieve probability $p_\text{low}$ for different node
availabilities $a$.\label{t:replicas}}{%
\centering
\begin{tabular}{|c|c|c|c|}\hline
{\bf Node availability} &\multicolumn{3}{c|}{\bf Number of replicas required} \\ \cline{2-4}
$a$ & $p_\text{low}=0.99$ & $p_\text{low}=0.98$ & $p_\text{low}=0.95$ \\ \hline
0.5 & 7 & 6 & 5\\ \hline
0.75 & 4 & 3 & 3\\ \hline
0.99 & 1 & 1 & 1 \\ \hline
\end{tabular}}
\end{table}

\begin{figure}
\centering
\subfloat[Storage efficient hybrid schemes (any $d$ value).]{\label{f:hybridr:msrlow}\includegraphics{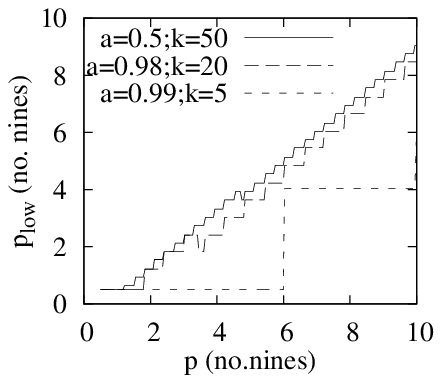}}
\subfloat[Bandwidth efficient hybrid schemes (when $d=k$).]{\label{f:hybridb:msrlow}\includegraphics{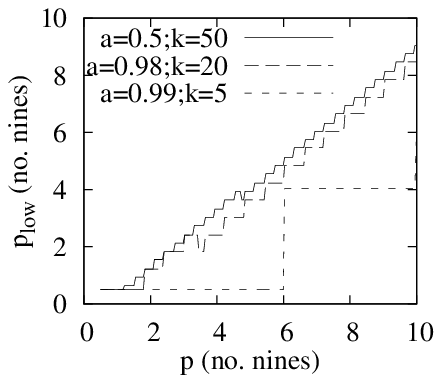}}
\subfloat[Bandwidth efficient hybrid schemes (when $d=n-1$).]{\label{f:hybridb:msrhigh}\includegraphics{simul/combining-bandwidth-msrlow.eps}}
\caption{ The $(p,p_\text{low})$-pairs under each of the lines represent the scenarios where a hybrid scheme
(replication+MSR codes) reduces the costs of a single replicated scheme. The lines are the maximum $p_\text{low}$ values
that satisfy eq.~(\ref{e:limitR}) for (a), and eq.~(\ref{e:limitW}) for (b) and (c).}
\label{f:hybrid}
\end{figure}

In Figure~\ref{f:hybridr:msrlow} we plot the maximum value for $p_\text{low}$ that satisfies eq.~(\ref{e:limitR}) as a
function of the overall retrieve probability $p$ for different on-line node availabilities $a$. The $k$ parameter is set
to $k=50$ when $a=0.5$, $k=20$ when $a=0.75$ and $k=5$ when $a=0.99$.  The $(p,p_\text{low})$-pairs below each of the
lines correspond to the hybrid instances that satisfy eq.~(\ref{e:limitR}), i.e.~a hybrid scheme reduces the storage
costs.  Similarly, in figures~\ref{f:hybridb:msrlow} and~\ref{f:hybridb:msrhigh}, we plot the $(p,p_\text{low})$-pairs
that satisfy~eq.(\ref{e:limitW}), i.e.~a hybrid scheme reduces the communication costs.

As example, let us assume a storage system that wants 99\% data availability for their replicated files.  In this case
($p_\text{low}=0.99$), looking at Figure~\ref{f:hybrid} we see that a hybrid scheme (replication+MSR codes) can reduce
the storage costs compared to  replication only  when $p\geq0.999999$ for $a=0.99$, when $p\geq0.99$ for $a=0.75$, and
when $p\geq0.9$ for $a=0.5$. Since in general we always want strong guarantees that files are never lost |e.g., we
assume $p\geq0.999999$|, we can conclude that \textit{hybrid schemes reduce storage and communication cost for almost
all practical scenarios}.

It is interesting to note that in Figure~\ref{f:hybrid} all three sub-figures look very much alike. The reason is that
the cost contribution of replication is significantly higher than the cost contribution of the coding (see
Section~\ref{s:costs}).  Since we have demonstrated the cost efficiency of a hybrid scheme for $p_\text{low}=0.99$,
which requires a larger number of replicas than configurations with $p_\text{low}\leq0.99$, see Table~\ref{t:replicas},
a hybrid scheme  will also reduce storage and communication costs for any system requiring fewer replicas i.e.,
$p_\text{low}\leq0.99$.

\section{Experimental Evaluation}
\label{s:experiment}

In previous sections we presented our generic storage model based on Regenerating Codes and we  analytically analyzed
the storage and communication costs for   MSR and MBR codes, as well as the efficiency of using these codes in hybrid
redundancy schemes. In this section, we aim to evaluate how the network traffic caused by repair processes can affect
the performance and scalability of the redundancy scheme. For that, we assume a distributed storage system constrained
by its network bandwidth: a system where storage nodes have low upload bandwidth and nodes have low on-line
availabilities.  For such a storage system we will evaluate two measures that are difficult to obtain analytically: (i)
the real bandwidth used by the  repair process |i.e.,~bandwidth utilization|, and (ii) the repair time |i.e.,~time
required to download $d$ fragments.  In this way we can evaluate the impact of the  repair degree   $d$  on bandwidth
utilization and  system scalability.

\paragraph{Bandwidth utilization} Given a node \textbf{upload bandwidth}, $\omega$, and the per-node required bandwidth,
$W$, we can theoretically state that a feasible storage system must satisfy $\omega\geq W$, and that the storage system
reaches its maximum capacity when $\omega=W$. However, practical storage systems may not reach this maximum capacity
because of system inefficiencies due to failed repairs or fragment retransmissions. To measure these inefficiencies, we
will compare the real bandwidth utilization  $\hat\rho$  with the theoretical bandwidth utilization $\rho=W/\omega$.

\paragraph{Repair time} The repair time is proportional to the repair bandwidth, $\gamma$, the repair degree, $d$, and
the probability $a$ of finding a node on-line~\cite{retrievaltimes}.  We showed in Section~\ref{s:costs} that increasing
$d$ reduces the repair bandwidth $\gamma$, (see eqs.~(\ref{e:parmsr}) and~(\ref{e:parmbr})), which  should then
intuitively reduce repair times.  However, since the system only guarantees $k$ on-line nodes, contacting   $d>k$ nodes
may require  to wait for nodes coming back on-line, which will cause longer repair times. In previous sections we only
considered two repair degrees $d$, namely $d=k$ and $d=n-1$. In this section we will analyze how different $d$ values
affect repair times and bandwidth utilization.

\subsection{Simulator Set-Up}
\label{s:setup}

We implemented an event-based simulator that simulates a dynamic storage infrastructure. Initially, the simulator starts with
$N=500$ storage nodes. New node arrivals follow a Poisson process with average inter-arrival times $\mathbb{E}[L]/N$.  Node
departures follow a Poisson process with the same inter-departure time. Once a node joins the system it draws its
lifetime from an exponential distribution $L$ with expected value $\mathbb{E}[L]=100$ days. During their lifetime in the
system, nodes alternate between  on-line/off-line sessions. For each session, each node draws its on-line and off-line
durations from distributions $\mathcal X_\text{on}$ and $\mathcal X_\text{off}$ respectively. In this paper $\mathcal
X_\text{on}$ and $\mathcal X_\text{on}$ are exponential variates with parameters $1/(B\cdot a)$ and $1/(B(1-a))$
respectively, where $B$ is the base time and $a$ the node on-line availability. Using the mean value of the exponential
distribution we can compute the average duration of the on-line and off-line periods as (in hours): \begin{align}
\mathbb{E}[\mathcal X_\text{on}]&=B\cdot a \\ \mathbb{E}[\mathcal X_\text{off}]&=B\cdot (1-a) \end{align}

The simulator implements  parameterized Regenerating Code. To cope with node failures, redundant blocks are repaired in
a proactive manner following the algorithm defined in~\cite{proactiveestim} and the simulator proactively generates new
redundant blocks at a constant rate. For each stored object, a new redundant block is generated every $\mathbb{E}[L]/n$
days. To balance the amount of data assigned to each node, each repair is assigned to the on-line node that is lest
loaded in terms of the number of stored blocks and the number of repairs going on.

If the repair node disconnects during a repair process, the repair is aborted and restarted at another on-line node.
Similarly, when a node uploading data disconnects, the partially uploaded data is discarded and the repair node  starts
a block retrieval from another on-line node.

The number of objects stored in the system is set in all the simulations to achieve a desired system utilization $\rho$.
Given  $\rho$, the number of stored objects, $O$, is obtained using the two following expressions:
\begin{align}
O_\text{MSR} &= \frac{\omega\rho ak(d-k+1)}{d\cdot \eta[k,a,p]} \times \frac{N\mathbb{E}[L]}{\mathcal M} \\
O_\text{MBR} &= \frac{\omega\rho ak(2d-k+1)}{2d\cdot \eta[k,a,p]} \times \frac{N\mathbb{E}[L]}{\mathcal M}
\end{align}
These formulas are obtained by taking the definition of
utilization, $\rho=W/\omega$, replacing $W$ by $\rho \cdot \omega$ in
eq.~(\ref{e:bw}) and solving the equation for $O$.

We set the on-line node availability to $a=0.75$  and   we set $k=20$.
With these values, we use eq.~(\ref{e:n}) to compute the minimum number of redundant blocks, $n$, required to
achieve a retrieve probability $p=0.999999$: $\eta[20,0.75,0.999999]=47$.

Finally, the node upload bandwidth is set to $\omega=$20KB/sec, allowing only one concurrent upload per node. To
simulate asymmetric network bandwidth, we allow up to 3 concurrent
downloads per node, which makes a maximum download
bandwidth of 60KB/sec.

\subsection{Impact of the Repair Degree $d$}

\begin{figure}
\centering
\subfloat[Bandwidth utilization]{\label{f:exp1:util}\includegraphics{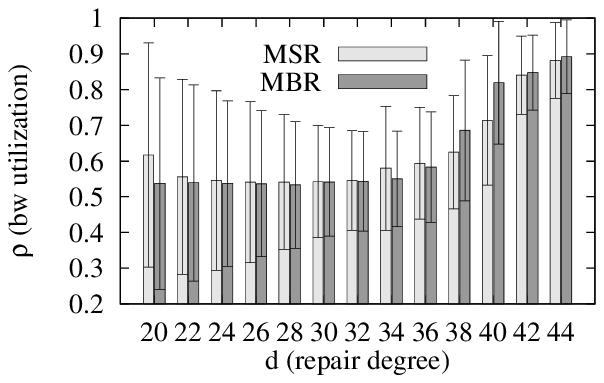}}
\subfloat[Repair times]{\label{f:exp1:time}\includegraphics{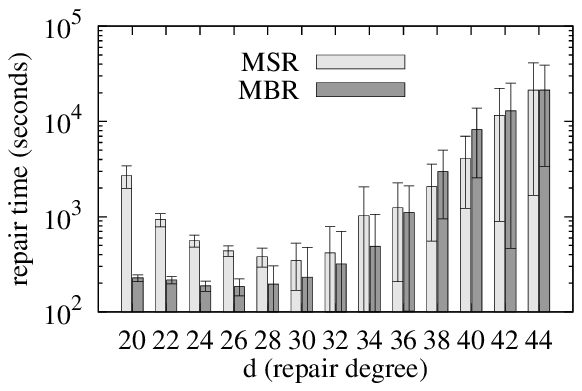}}\quad
\quad
\subfloat[Number of Objects]{\label{f:exp1:objs}\includegraphics{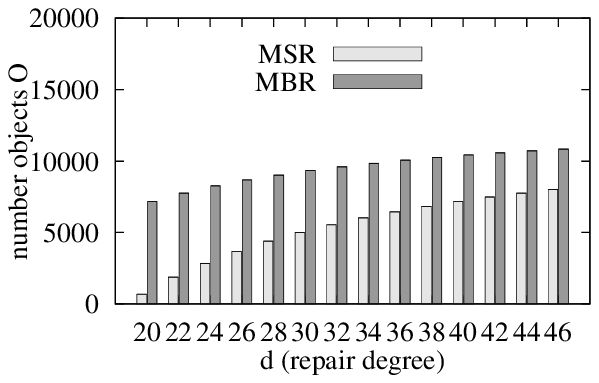}}
\subfloat[Overall Disk Utilization ($O\times\mathcal M\times R$)]{\label{f:exp1:red}\includegraphics{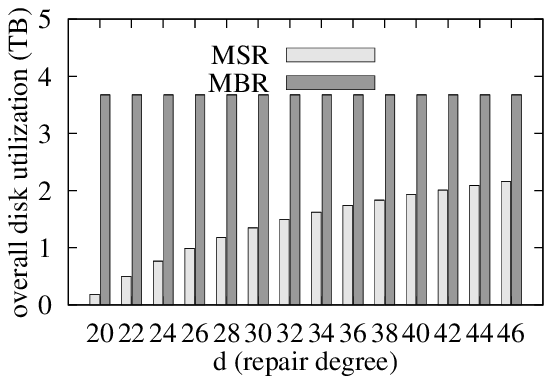}}
\caption{Bandwidth utilization and repair times for MSR and MBR and different repair degrees $d$ when the object size is
$\mathcal M=$120MB and the number of objects $O$ is set to achieve half bandwidth utilization $\rho=0.5$. The rest of the
parameters are set to: $k=20$ and $B=24$hours.}
\label{f:exp1}
\end{figure}

In Figure~\ref{f:exp1} we measure the effect of the repair degree   on the system utilization and on the repair
times. In this experiment, we set the size of the object to $\mathcal M=120$MB and the base time to $B=24$ hours
|i.e.~on average nodes connect and disconnect once per day. The number of stored objects is set to achieve a
bandwidth utilization of $\rho=0.5$. Figure~\ref{f:exp1:objs} shows
the number of objects $O$ for $\rho=0.5$, and
Figure~\ref{f:exp1:red}  the storage space required.
 Figures~\ref{f:exp1:util} and~\ref{f:exp1:time} show that small $d$
 values (values close to $k=20$) allow to keep the bandwidth
 utilization on target and assure low repair times. However, for repair degrees $d>34$ the repair times
start to increase exponentially.

It is interesting to see that when the repair times are quite long, nodes executing repairs may not finish their repairs
before disconnecting since repair times become longer than on-line sessions. In this case, failed repairs are reallocated
and restarted in other on-line nodes. These unsuccessful repairs cause
useless traffic that increase then the real bandwidth
utilization. In Figure~\ref{f:exp1:util} we can see how for $d>38$ repair times start to be larger than on-line
sessions, increasing utilization beyond $0.5$. It is important to note that these larger repair times can jeopardize the
reliability of the system: \emph{large $d$ values can cause most repairs to fail, reducing the amount of available
blocks and reducing the probability of successfully accessing stored files}.

\begin{figure}
\centering
\subfloat[Bandwidth utilization]{\label{f:exp3:util}\includegraphics{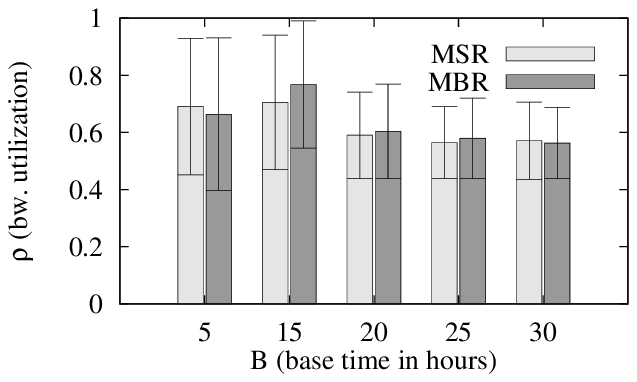}}
\subfloat[Repair times]{\label{f:exp3:time}\includegraphics{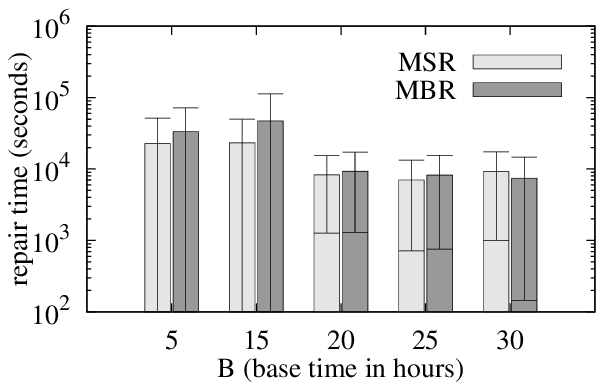}}
\caption{Bandwidth utilization and repair times for MSR and MBR and different base times $B$ when the object size is
$\mathcal M=$120MB and the number of objects $O$ is set to achieve a bandwidth utilization $\rho=0.5$. The rest of the
parameters are set to: $k=20$ and $d=36$. For the MSR case $O=5069$, and for MBR $O=10984$.}
\label{f:exp3}
\end{figure}

To investigate the increase of bandwidth utilization in detail, we analyze in Figure~\ref{f:exp3} the performance of the
storage system for the point where repair times begin to increase, $d=36$. At this point we evaluate repair times and
bandwidth utilization for different base times, $B$. As $B$ increases, the duration of on-line sessions become longer
and fewer repairs need to be restarted, theoretically reducing bandwidth utilization. We can see this effect in
Figure~\ref{f:exp3:util}, larger base times reduce the bandwidth utilization of the system. Due to this utilization
reduction, repair times are also slightly reduced as we can see in Figure~\ref{f:exp3:time}.

\subsection{Scalability}

\begin{figure}
\centering
\subfloat[Bandwidth utilization]{\includegraphics{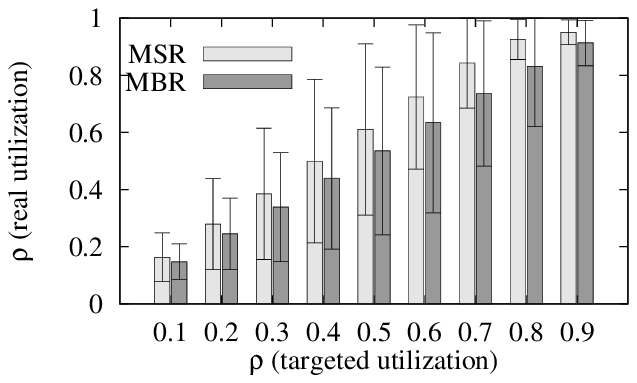}}
\subfloat[Repair times]{\includegraphics{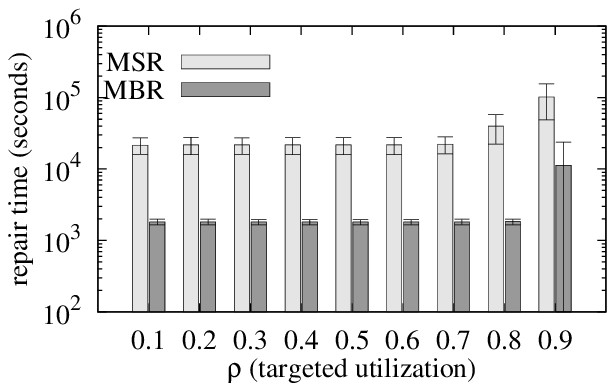}}
\quad
\subfloat[Number of Objects]{\label{f:exp2:objs}\includegraphics{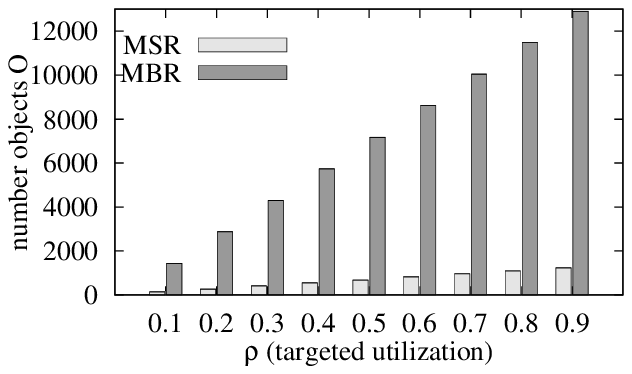}}
\caption{Bandwidth utilization and repair times for MSR and MBR and different targeted utilizations $\rho$ when the
object size is $\mathcal M=$120 MB and the number of objects $O$ is set to achieve the targeted $\rho$. The rest of the
parameters are set to: $k=20$, $B=24$ hours and $d=20$.}
\label{f:exp2}
\end{figure}

Other than the impact of the \emph{repair degree} $d$ and the base
time $B$ we aim to analyze the behavior of the storage system under
different target bandwidth utilizations. In Figure~\ref{f:exp2} we plot the
measured utilization and repair times for a wide range of target
utilizations $\rho$. We set the size of the stored objects to 120MB
and we increase the number of stored objects, $O$, to achieve
different utilizations. In this scenario we set $k=d=20$. In
Figure~\ref{f:exp2}a we  see how the measured utilization is nearly
the same than the target utilization. This is because $d=k$ causes
short repair times and repairs typically finish before nodes go off-line. However,
in Figure~\ref{f:exp2}b we can appreciate how for a high bandwidth
utilization of $\rho=0.9$, the saturation of the node upload queues
increases repair times significantly.

\begin{figure}
\centering
\subfloat[Bandwidth utilization]{\includegraphics{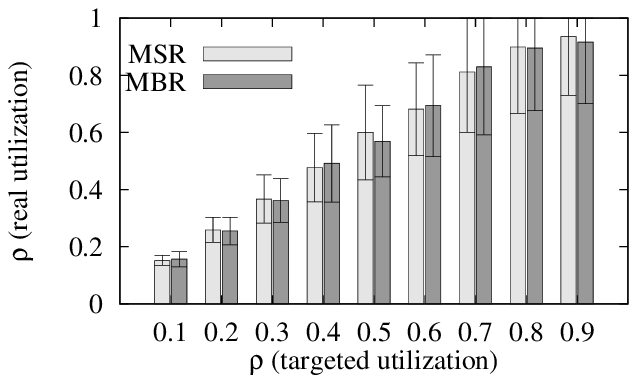}}
\subfloat[Repair times]{\includegraphics{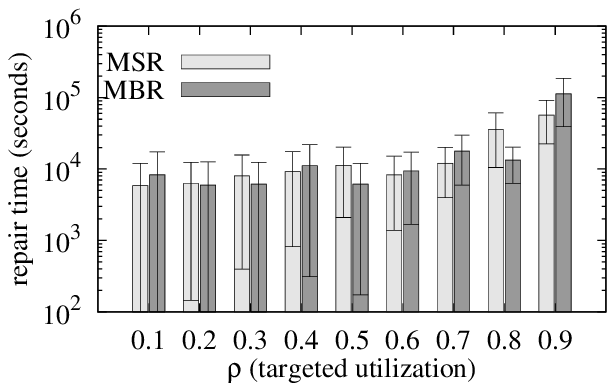}}
\quad
\subfloat[Number of Objects]{\label{f:exp2a:objs}\includegraphics{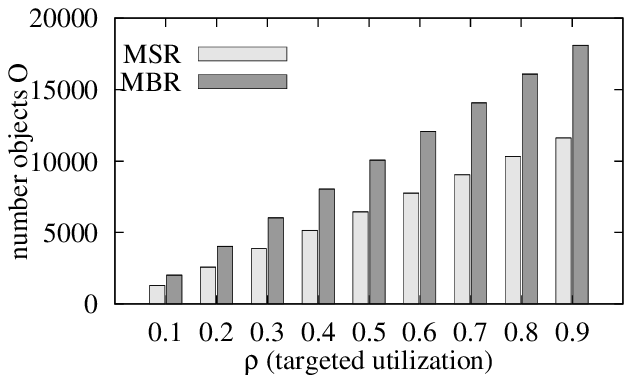}}
\caption{Bandwidth utilization and repair times for MSR and MBR and different targeted utilizations $\rho$ when the
object size is $\mathcal M=$120MB and the number of objects $O$ is set to achieve the targeted $\rho$. The rest of the
parameters are set to: $k=20$, $B=24$ hours and $d=36$.}
\label{f:exp2a}
\end{figure}

In Figure~\ref{f:exp2a} we plot the same metrics as in
Figure~\ref{f:exp2} but for a repair degree of $d=36$.
Increasing the repair degree causes longer retrieval times, however as
we saw in Figure~\ref{f:exp1},  
$d=36$ keep repairs short enough to guarantee that the utilization is
not affected. However,  by
increasing the repair degree from $d=20$ to $d=36$ we 
can store on the same system configuration  one order of magnitude
more objects, namely 6452 (MSR, $d=36$) instead of 
 683 (MSR, $d=20$).

\begin{figure}
\centering
\subfloat[Bandwidth utilization]{\includegraphics{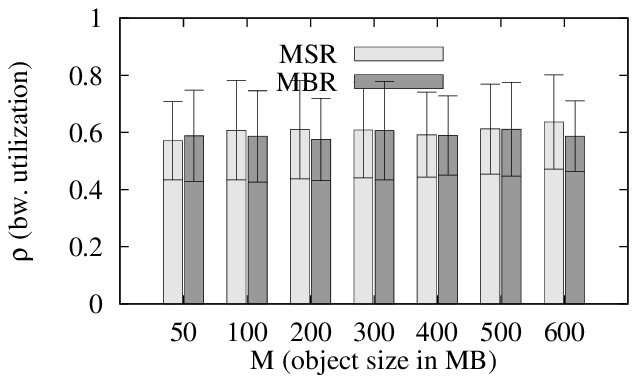}}
\subfloat[Repair times]{\includegraphics{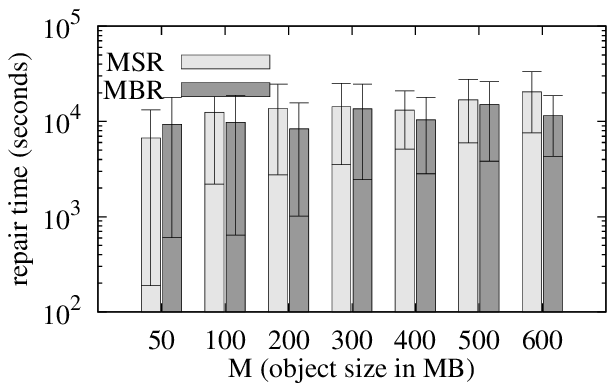}}
\quad
\subfloat[Number of Objects]{\label{f:exp4:objs}\includegraphics{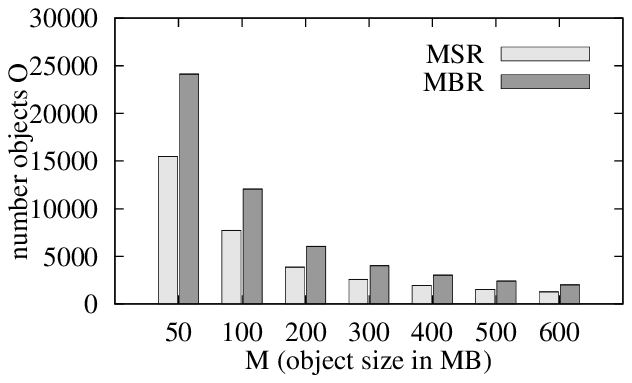}}
\caption{Bandwidth utilization and repair times for MSR and MBR and different object sizes $\mathcal M$ when the
number of objects $O$ is set to achieve a bandwidth utilization $\rho=0.5$. The rest of the
parameters are set to: $k=20$, $B=24$hours and $d=36$.}
\label{f:exp4}
\end{figure}

Finally, in Figure~\ref{f:exp4} we analyze the impact of object size
on bandwidth utilization and repair times.  For each object size we
set the number of stored objects to achieve a target bandwidth
utilization of $\rho=0.5$.  Since the utilization is the same for all
object sizes, the number stored objects, $O$, decreases as the object
size increases (Figure~\ref{f:exp4}c). Independently of the object
size, the total amount of stored data, $O\times\mathcal M$ remains
constant: 774GB for MSR codes and 1206GB for MBR codes. We can also
see in Figure~\ref{f:exp4}a that the measured bandwidth utilization is
independent of the object size. However, as expected, we can see in
Figure~\ref{f:exp4}b that larger objects take longer to repair.

\section{Conclusions}
\label{s:conclusions}

In this paper we evaluated redundancy schemes for distributed storage
systems in order   to have a clearer understanding of the  cost trade-offs in distributed storage systems.
Specifically, we analyzed the performance of the generic family of erasure codes called Regenerating
Codes~\cite{regeneratingcodes}, and the use of Regenerating Codes in
hybrid redundancy schemes. For each parameter
combination   we analytically derived its storage and communication costs  of
Regenerating Codes.
Our cost analysis is novel
in that it takes into account the effects of on-line node availabilities and node lifetimes. Additionally, we used
an event-based simulator to evaluate the effects of network utilization on the scalability of different redundancy
configurations. Our main results are as follows:

\begin{itemize}
\item Compared to simple replication, the use of a Regenerating Codes can reduce the costs of a storage system (storage
and communication costs) from 20\% up to 80\%.
\item The optimal value of the retrieval degree $k$ depends on the on-line node availability, ranging from $k=5$ when
nodes have 99\% availability, to $k=50$ when nodes have 50\% availability. Once $k$ is fixed, storage systems with
limited storage capacity can maximize their storage capacity by adopting MSR codes. On the other hand, systems with
limited communications bandwidth can maximize their storage capacity by adopting MBR codes.
\item High repair degrees $d$ reduce the overall communication costs
  but may increase repair times significantly, which can lead to data loss. 
We experimentally found that the repair degree should be small enough to make sure the repair times are shorter than the
on-line session durations of nodes.
\item Finally, in storage systems where the access to whole file replicas is required, we showed that hybrid schemes
combining replication and MSR codes are more cost efficient than simple replication.
\end{itemize}

\begin{acks}
This work was done during a visit of the first author at Eurecom in Spring 2010.
We would like to thank Matteo Dell'Amico and Zhen Huang for their helpful comments and their support of this research.
\end{acks}

\bibliographystyle{acmsmall}
\bibliography{citations}

\end{document}